\documentclass[11pt,letterpaper]{article}
\usepackage{color}
\usepackage{amsmath}
\usepackage{amsthm}
\usepackage{amsfonts}
\usepackage{amssymb}
\usepackage{graphicx}

\usepackage[all]{xy}
\usepackage{amscd}
\usepackage{amsthm}
\usepackage{longtable}
\usepackage{mathrsfs}
\usepackage{amsmath}
\usepackage{amssymb}
\usepackage{enumitem}

\usepackage[utf8]{inputenc}
\usepackage[english]{babel}
 
\newtheorem{theorem}{Theorem}[section]
\newtheorem{corollary}[theorem]{Corollary}
\newtheorem{lemma}[theorem]{Lemma}
\theoremstyle{definition}
\newtheorem{definition}{Definition}[section]

\newtheorem*{remark}{Remark}

\newtheorem*{counter example}{Counter Example}

\begin{document}
\title{Rotation Symmetric Bent Boolean Functions for $n=2p$}
\author{T. W. Cusick\\
	 University at Buffalo, 244 Math. Bldg. Buffalo, NY 14260\\
	 Corresponding author- e-mail: cusick@buffalo.edu\\
\and E. M. Sanger\\ 
University at Buffalo, 244 Math. Bldg., Buffalo, NY 14260\\
 e-mail: ellynsan@buffalo.edu} 
\date{}
\maketitle
	
\begin{abstract}
It has been conjectured that there are no homogeneous rotation symmetric bent Boolean functions of degree greater than two.  In this paper we begin by proving that sums of short-cycle rotation symmetric bent Boolean functions must contain a specific degree two monomial rotation symmetric Boolean function.  We then prove most cases of the conjecture in $n=2p$, $p>2$ prime, variables and extend this work to the nonhomogeneous case. 
\end{abstract}

\section{Introduction}
\label{intro}
A Boolean function in $n$ variables can be defined as a map from $\mathbb{V}_n$, the $n$-dimensional vector space over the two element field $\mathbb{F}_2,$ to $\mathbb{F}_2$.
  If $f$ is a Boolean function in $n$ variables, the {\it truth table} of $f$ is defined to be the $2^n$-tuple given by $(f({\bf v}_0),f({\bf v}_1),\ldots ,f({\bf v}_{2^n-1}))$ where ${\bf  v}_0=(0,\ldots ,0,0),{\bf  v}_1=(0,\ldots ,0,1),\ldots ,$ ${\bf  v}_{2^n-1}=(1,\ldots ,1,1)$ are the $2^n$ elements of $\mathbb{V}_n$ listed in lexicographical order. The {\it weight} or {\it Hamming weight} of $f$ (notation  $wt(f)$) is 
the number of $1$'s that appear in the truth table of $f$.\\
\indent As described in \cite[pp. 5-6]{CBF}, every Boolean function on $\mathbb{V}_n$ can be expressed as a polynomial over $\mathbb{F}_2$ in $n$ binary variables by:
$$ f(x_0,\ldots ,x_{n-1})=\displaystyle{\sum_{{\bf a}\in \mathbb{V}_n}}c_{\bf a}x_0^{a_0}\cdots x_{n-1}^{a_{n-1}} $$
where $c_{\bf a}\in \mathbb{F}_2$ and ${\bf a}=(a_0,\ldots ,a_{n-1})$ with each $a_i$ equal to $0$ or $1.$ The above representation is  the {\it algebraic normal form} (ANF) of $f$. Let $d_i$ be the number of variables in the $i$-th monomial of $f$, so $d_i$ is the {\it algebraic degree} (or just the {\it degree}) of the monomial. If we let $D$ be the set of the distinct degrees of the monomials in $f$ which have non-zero coefficients, then the {\it degree} (notation $\deg(f)$) of $f$ is given by max$(D)$. If $D$ contains only one element, then each monomial in $f$ has the same degree and $f$ is {\it homogeneous}. If $\deg(f)=1$, then $f$ is {\it affine}, and if $f$ is affine and homogeneous (i.e. the constant term is 0), then $f$ is {\it linear}.   \\
\indent  A Boolean function $f$ is {\it rotation symmetric} (RotS) if its ANF is invariant under any power of the cyclic permutation $\rho(x_0,\ldots ,x_{n-1})=(x_1,\ldots,x_{n-1},x_0)$.  We will use the notation $u\sim v$ to indicate that there exists some $1\leq k\leq n$ such that $\rho^k(u)=v$.   Clearly $\sim$ defines an equivalence relation on  
$\mathbb{V}_n.$

Let $O_n({\bf x})$ denote the orbit of {\bf x} under the action of $\rho^k$, $1\leq k\leq n$, and let $\mathcal{G}_n$  be the set of representatives of all the orbits in 
$\mathbb{V}_n$.   Then a rotation symmetric Boolean function $f$ can be written as 
$$a_0\oplus a_1x_0\oplus\sum\limits_{j=1}^{n-1}a_{1j}x_0x_j\oplus\ldots\oplus a_{12\ldots n}x_0x_1\cdots x_{n-1},$$ 
where the coefficients $a_0,a_1,a_{1j},\ldots,a_{12\ldots n}\in\mathbb{F}_2$, and the existence of a representative term $x_0x_{i_2}\cdots x_{i_l}$ implies the existence of all the terms from $O_n(x_0,x_{i_2},\ldots,x_{i_l})$ in the algebraic normal form.  We call this representation the {\it short algebraic normal form} (SANF) of $f$.  

Unless otherwise specified, all subscripts in any monomial will be taken mod~($n$) with entries in $\{0, 1, \ldots, n-1\}.$ We may omit the modulus if it is clear from the context.

Suppose $f$ has SANF $x_0x_{i_1}\ldots x_{i_l}$, then we say $f$ is  a {\it monomial rotation symmetric} (MRS) function and call it  {\it full-cycle} if $|O_n(x_0,x_{i_1},\ldots,x_{i_l})|=n$.  Thus if $f$ is full-cycle then it contains $n$ monomials in its ANF.  If $|O_n(x_0,x_{i_1},\ldots,x_{i_l})|<n$ we say $f$ is {\it short-cycle}.  In this case $f$ will contain $n/w$ monomials in its ANF for some divisor $w > 1$ of $n.$

The {\it Hamming distance} between two Boolean functions $f$ and $g$, denoted $d(f,g)$, is defined as $d(f,g)=wt(f\oplus g)$.  Each Boolean function $f$ has an associated {\it sign function}, $\hat{f}:\mathbb{R}^*\rightarrow\mathbb{C}^*$, defined by $\hat{f}({\bf x})=(-1)^{f({\bf x})}$. 

\begin{definition} 
The {\it Walsh transform} of a function $f:\mathbb{V}_n\rightarrow\mathbb{F}_2$ is the map $W(f):\mathbb{V}_n\rightarrow\mathbb{R}$ given by, 
\begin{equation*}
W(f)({\bf w})=\sum\limits_{{\bf x}\in\mathbb{V}_n}f({\bf x})(-1)^{{\bf w}\cdot{\bf x}}
\end{equation*}
\end{definition} 

The {\it Walsh spectrum} of $f$ is the list of the $2^n$ {\it Walsh coefficients} given by $W(f)({\bf w})$ as {\bf w} varies.

\begin{definition}
A matrix $H$ is called a {\it Hadamard matrix} of order $n$ if it is an $n\times n$ matrix of $\pm1$s such that 
\begin{equation*}
HH^t=nI_n,
\end{equation*}
where $H^t$ is the transpose of $H$ and $I_n$ is the $n\times n$ identity matrix. 
\end{definition}

\begin{definition}
A Boolean function $f$ in $n$ variables is called {\it bent} if and only if the Walsh transform coefficients of $\hat{f}$ are all $\pm 2^{n/2}$. 
\end{definition}

\begin{definition}
The {\it nonlinearity} of a function $f,$ denoted by $\mathcal{N}_f,$ is defined by
\[
\mathcal{N}_f=\min_{g \in {\cal A}_n} d(f, g),
\]
where ${\mathcal{A}}_n$ is the set of all affine functions on $\mathbb{V}_n$.
\end{definition}  

From \cite[pp. 76-77]{CBF} we have the following equivalent definition of the bent property:

\begin{theorem}
\label{Had}
Let $f:\mathbb{V}_n\rightarrow\mathbb{F}_2$ be a Boolean function.  Then the following are equivalent:
\begin{itemize}
\item[(i)] $f$ is bent.
\item[(ii)] $M_{\hat{f}}=(\hat{f}({\bf u}\oplus{\bf v}))$ is a Hadamard matrix.
\item[(iii)] The nonlinearity of $f$ is $\mathcal{N}_f=2^{n-1}-2^{n/2-1}.$
\end{itemize}
\end{theorem}

It is easy to observe that bent functions exist only for even dimensions.  Let $n$ be even and ${\bf v}=(a_0,a_1,\ldots,a_{n-1})$; then we will use the notation 
${\bf v}'=(a_0,\ldots,a_{\frac{n}{2}-1})$ and ${\bf v}''=(a_{\frac{n}{2}},\ldots,a_{n-1})$. 

\begin{lemma}
\label{degn/2}
For $n=2$, the degree of a bent function on $\mathbb{V}_2$ is $2$.  For $n>2$, the degree of a bent function is at most $n/2$.
\end{lemma}
\begin{proof}
The proof of this can be found in \cite[pp. 80]{CBF}.
\end{proof}

\section{Rotation symmetric bent functions when n=2m}
\label{HomRSany} 

In this section we will consider rotation symmetric functions of any even degree $n$.  We will first prove that any degree $2$ rotation symmetric bent function must contain $f_0=\bigoplus\limits_{i=0}^{n/2-1}x_ix_{\frac{n}{2}+i}$ and then extend this to the case where $f$ is composed of short-cycle MRS functions.  We begin by proving the well known fact that $f_0$ is bent.

\begin{lemma}
\label{f0bent}
Let $n=2m$ and $f_0=\bigoplus\limits_{i=0}^{m-1}x_ix_{m+i}$.  Then $wt(f_0)=2^{2m-1}-2^{m-1}$ and $f_0$ is bent.
\end{lemma}
\begin{proof}
The fact that $f$ is bent follows from \cite[Cor. 5.23, p. 82]{CBF}, after a permutation of the variables.  Then a computation gives the weight.
\end{proof} 

To prove that any degree $2$ rotation symmetric bent function, $f$, must contain $f_0$, we will look at the matrix $M=(\hat{f}({\bf v}_i\oplus{\bf v}_j))(\hat{f}({\bf v}_i\oplus{\bf v}_j))^T$  to determine when $(\hat{f}({\bf v}_i\oplus{\bf v}_j))$ is a Hadamard matrix. We will first simplify the matrix $M$  to a useful form.

\begin{lemma}
Let $M=(\hat{f}({\bf v}_i\oplus{\bf v}_j))(\hat{f}({\bf v}_i\oplus{\bf v}_j))^T$ where $f$ is a Boolean function in $n$ variables, then 
\begin{equation}
\label{[1]}
M_{i,j}=\sum_{k=0}^{2^n-1}\hat{f}({\bf v}_i\oplus{\bf v}_{k})\hat{f}({\bf v}_{k}\oplus{\bf v}_j). 
\end{equation}
\end{lemma}
\begin{proof}
Let $A=(\hat{f}({\bf v}_i\oplus{\bf v}_j))$.  Notice $A^{T}=A$, so $M=AA$.  Then 
\begin{equation*}
M_{i,j}=(AA)_{i,j}=\sum\limits_{k=0}^{2^n-1}A_{ik}A_{kj}=\sum\limits_{k=0}^{2^n-1}\hat{f}({\bf v}_i\oplus{\bf v}_k)\hat{f}({\bf v}_k\oplus{\bf v}_j).
\end{equation*} 
\end{proof}

\begin{lemma}
\label{sim}  
We have ${\bf v}_i\sim{\bf v}_j$ if and only if there exists $k, 1\leq k\leq n$ such that $2^ki \equiv j \bmod(2^n-1).$
\end{lemma}
\begin{proof}
Let ${\bf v}_i=(a_0,a_1,...,a_{n-2},a_{n-1}),$ so $i=\displaystyle\sum\limits_{u=1}^{n}a_{n-u}2^{u-1}.$
\noindent
$\Rightarrow$: Assume ${\bf v}_i\sim{\bf v}_j$, then there exists $\rho^l,\ 1\leq l\leq n,$ such that 
\begin{equation*}
\rho^l({\bf v}_i)=(\rho^l(a_{0}),\ldots,\rho^l(a_{n-2}),\rho^l(a_{n-1}))=(a_{l},\ldots,a_{l+n-2},a_{l+n-1})={\bf v}_j.
\end{equation*}
Thus, $j=\displaystyle\sum\limits_{u=1}^{n}a_{l+n-u}2^{u-1}$
where $1\leq l\leq n$ and 
\begin{eqnarray*}
2^{l}i &=&\displaystyle\sum\limits_{u=1}^{n}a_{n-u}2^{l+u-1}\\ 
&=&a_{n-1}2^{l}+\ldots+a_{n-(n-l)}2^{l+(n-l)-1}+a_{n-(n-l+1)}2^{l+(n-l+1)-1}+\\
&&a_{n-(n-l+2)}2^{l+(n-l+2)-1}+\ldots+a_{0}2^{l+n-1}\\
              &=& a_{n-1}2^{l}+\ldots+a_{l}2^{n-1}+a_{l-1}2^{n}+a_{l-2}2^{n+1}\ldots+a_{0}2^{n-1+l}\\
              &\equiv& \mbox{$a_{l-1}+a_{l-2}2+\ldots+a_{l}2^{n-1}$ mod~($2^n-1$)}\\
              &\equiv& \mbox{$a_{l+n-1}+a_{l+n-2}2+\ldots+a_{l+1}2^{n-2}+a_{l}2^{n-1}$ mod~($2^n-1$)}\\
              &\equiv& \mbox{$j$ mod~($2^n-1$)},
\end{eqnarray*}
since indices are taken mod $n$.

Thus there exists $k, 1\leq k\leq n$ such that $2^ki\equiv$\mbox{$j$ mod~($2^n-1$)}.\\\\
\noindent
$\Leftarrow$: Assume there exists $k, 1\leq k\leq n$ such that $2^ki\equiv j \mod{(2^n-1)}$.  Then
\begin{eqnarray*}
2^ki &=&\displaystyle\sum\limits_{u=1}^na_{n-u}2^{k+u-1}\\ 
	&=&a_{n-1}2^k+\ldots+a_{n-(n-k)}2^{k+(n-k)-1}+a_{n-(n-k+1)}2^{k+(n-k+1)-1}\\
	&&+a_{n-(n-k+2)}2^{k+(n-k+2)-1}+\ldots+a_{0}2^{n-1+k}\\
	&=&a_{n-1}2^k+\ldots+a_k2^{n-1}+a_{k-1}2^{n}+a_{k-2}2^{n+1}+\ldots+a_02^{n-1+k}\\
        &\equiv& a_{k-1}+a_{k-2}2+\ldots+a_{k}2^{n-1} \mod{(2^n-1)}\\
        &=&j.
\end{eqnarray*}
Thus ${\bf v}_j=(a_{k-1},a_{k-2},\ldots,a_k)$ so $\rho^{k-1}({\bf v}_i)={\bf v}_j$ and ${\bf v}_i\sim{\bf v}_j.$
\end{proof}

\begin{lemma} 
\label{conjeq}
We have ${\bf v}_a\sim{\bf v}_b$ if and only if  ${\bf v}_{2^n-1-a}\sim{\bf v}_{2^n-1-b}$.
\end{lemma}
\begin{proof}  By Lemma \ref{sim},
\begin{eqnarray*}
{\bf v}_a\sim{\bf v}_b &\iff& \text{there exists}~ k, 1\leq k\leq n\text{ such that }2^ka\equiv \mbox{$b$ mod~($2^n-1$)}\\
                                     &\iff& 2^n-1-2^ka\equiv 2^n-1-b\\
                                     &\iff& 2^k(2^n-1-a)\equiv 2^n-1-b\\
                                     &\iff& {\bf v}_{2^n-1-a}\sim{\bf v}_{2^n-1-b}.
\end{eqnarray*}
\end{proof}

\begin{lemma}
\begin{equation}
\label{[2]}
M_{0,2^n-1}=\sum_{k=0}^{2^{n-1}-1}2\hat{f}({\bf v}_k)\hat{f}({\bf v}_{2^n-1-k}).
\end{equation}
\end{lemma}
\begin{proof}
From \eqref{[1]} we have,
\begin{eqnarray*}
M_{0,2^n-1} &=&\sum\limits_{k=0}^{2^n-1}\hat{f}({\bf v}_0\oplus{\bf v}_{k})\hat{f}({\bf v}_k\oplus{\bf v}_{2^n-1})=\sum
\limits_{k=0}^{2^n-1}\hat{f}({\bf v}_k)\hat{f}({\bf v}_{2^n-1-k})\\
&=&\hat{f}({\bf v}_0)\hat{f}({\bf v}_{2^n-1})+\hat{f}
({\bf v}_1)\hat{f}({\bf v}_{2^n-2})+\ldots+\hat{f}({\bf v}_{2^{n-1}-1})\hat{f}({\bf v}_{2^{n-1}})\\
& &+\hat{f}({\bf v}_{2^{n-1}})\hat{f}
({\bf v}_{2^{n-1}-1})+\ldots+\hat{f}({\bf v}_{2^n-2})\hat{f}({\bf v}_1)+\hat{f}({\bf v}_{2^n-1})\hat{f}({\bf v}_0)\\
&=&\sum\limits_{k=0}
^{2^{n-1}-1}2\hat{f}({\bf v}_k)\hat{f}({\bf v}_{2^n-1-k}).
\end{eqnarray*}
\end{proof} 


\begin{lemma}
\label{f0onlyshort}
The function $f_0=\bigoplus\limits_{i=0}^{m-1}x_ix_{m+i}$ is the only MRS short-cycle function of degree $2$ in $n=2m$ variables.
\end{lemma}
\begin{proof}
Let $f$ be an MRS function with SANF $x_0x_k$.  Suppose $f$ is short-cycle, then $f$ has $1\leq M<2m$ monomials and 
\begin{equation*}
f=x_0x_k\oplus x_1x_{k+1}\oplus\ldots\oplus x_{M-1}x_{k+M-1}. 
\end{equation*}
Thus $\{0,k\}=\{M,k+M\}$.  Since $1\leq M<2m=n$ then we have the following cases:
\begin{itemize}
\item[(i)] $0=M$ and $k\equiv\mbox{$M+k$ mod~($n$)}$. This is a contradiction since $1\leq M$.
\item[(ii)] $k=M$ and $0\equiv\mbox{$M+k$ mod~($n$)}$. Thus $0\equiv\mbox{$M+k$ mod~($n$)}\equiv\mbox{$2M$ mod~($2m$)}$.  So $2M=2am$ for some $a\in\mathbb{Z}$. Since $1\leq M<2m,$ then $2\leq 2am<4m$ and so $a=1$.  Thus $M=m$. 
\end{itemize}
Thus $f$ is a short-cycle if and only if $k=m$ which gives $f=f_0$.
\end{proof}

We now have enough to prove that $f_0$ is the only homogeneous MRS bent function of degree $2$:

\begin{theorem}
\label{f0only2}
Let $f$ be a rotation symmetric boolean function in $n=2m$ variables which has the SANF $x_0x_k$.  Then $f$ is bent if and only if $k=m$.
\end{theorem}
\begin{proof}
Assume $k\neq m=\frac{n}{2}$.  Then by Lemma \ref{f0onlyshort}, $f$ is full cycle and we write $f=\bigoplus\limits_{i=0}^{n-1}x_ix_{k+i}$.  Then each variable  $x_i,$  $0 \leq i \leq n-1,$  appears in two distinct monomials in $f$ and there are $n$ (even) monomials in $f$.  We will first show $f({\bf v}_i)=f({\bf v}_{2^n-1-i}),~0\leq i\leq n-1,$ and hence $\hat{f}({\bf v}_i)=\hat{f}({\bf v}_{2^n-1-i})$.  Notice $f({\bf v}_i)=0$ when there are an even number of monomials in $f$ which have a value of $1$ and $f({\bf v}_i)=1$ when there are an odd number of monomials which have a value of $1$.  

Fix $i=I$ and let $wt({\bf v}_I)=l.$ 
Suppose $f({\bf v}_I)$ has $r,s,\text{ and }t$ monomials of the form $x_jx_h$ with $\{x_j,x_h\}=\{1,1\},\{1,0\},\text{ and } \{0,0\}$ respectively, so $f({\bf v}_I) \equiv r \bmod{2}$.  Since each variable  $x_i$  appears in two distinct monomials of $f$, then by counting the number of times  $x_i=1$  in $f$ we see $2r+s=2l$ and $s$ must be even.  Also, we know there are $n$ monomials in $f$ so $r+s+t=n$.  Clearly $f({\bf v}_I)=0 \iff$ $r$ is even $\iff$ $t$ is even, since $s,~n$ are both even and $r+s+t=n$.  It is easy to see that $f({\bf v}_{2^n-1-I})=f(\bar{{\bf v}_I}) \equiv t \bmod{2}$.  Thus $f({\bf v}_{2^n-1-I})=0\iff\ t$ is even.
 
Hence for any $i,~ 0 \leq i \leq 2^n-1,$ we have $f({\bf v}_i)=f({\bf v}_{2^n-1-i})$ and therefore $\hat{f}({\bf v}_i)=\hat{f}({\bf v}_{2^n-1-i})$.  So \eqref{[2]} becomes:
\begin{equation*}
M_{0,2^n-1}=\sum\limits_{i=0}^{2^{n-1}-1}2(\hat{f}({\bf v}_i))^2=2^n.
\end{equation*}
Thus if $k\neq m$ then $M_{0,2^n-1}\neq 0$ and we have that $(\hat{f}({\bf v}_i\oplus{\bf v}_j))$ is not Hadamard.  Thus by Theorem \ref{Had} $f$ is not bent.  The remainder of the theorem follows from Lemma \ref{f0bent}.
\end{proof}

\begin{corollary}
\label{2f0contain}
If $n=2m$, then any degree $2$ rotation symmetric bent function must contain $f_0=\bigoplus\limits_{i=0}^{m-1}x_ix_{m+i}$.
\end{corollary}
\begin{proof}
Assume $f$ is a rotation symmetric function of degree $2$.  By \cite[Th. 5.29, p. 83]{CBF} we can assume $f$ has only quadratic monomials, so $f$ has SANF 
$\bigoplus\limits_{i=1}^{n-1}a_{i}x_0x_i$ where the coefficients  $a_{i} \in \mathbb{F}_2$.  Let $g_k$ be the function with SANF $a_{k}x_0x_k$ where $1\leq k \leq n-1.$  Suppose $f$ does not contain $f_0$, then $a_{\frac{n}{2}}=0$.  We have $g_k= a_k \bigoplus\limits_{i=0}^{n-1}x_ix_{k+i}$ when $1\leq k\leq n-1$ and from the proof of Theorem \ref{f0only2},
since $a_{n/2}=0$  we know $\hat{g_k}({\bf v}_i)=\hat{g_k}({\bf v}_{2^n-1-i})$ for $1 \leq i \leq n-1$.  Also, $\hat{f}=\widehat{\bigoplus\limits_{k=1}^{n-1}g_k}=\prod\limits_{k=1}^{n-1}\hat{g_k}$.  Thus 
$\hat{f}({\bf v}_i)=\prod\limits_{k=1}^{n-1}\hat{g_k}({\bf v}_i)=\prod\limits_{k=1}^{n-1}\hat{g_k}({\bf v}_{2^n-1-i})=\hat{f}({\bf v}_{2^n-1-i})$.  So \eqref{[2]} gives:
\begin{equation*}
M_{0,2^n-1}=\sum\limits_{k=0}^{2^{n-1}-1}2(\hat{f}({\bf v}_k))^2=2^n.
\end{equation*}
Thus $M_{0,2^n-1}\neq 0$, so $(\hat{f}({\bf v}_i\oplus{\bf v}_j))$ is not Hadamard and therefore by Theorem \ref{Had} $f$ is not bent.
\end{proof} 

Necessary and sufficient conditions for a bent quadratic function are given in \cite[Lemma 1, p. 4909]{Gao12}.  The Lemma gives an alternate proof of Corollary \ref{2f0contain} as discussed in 
 \cite[Rmk. 1, p. 4910]{Gao12}.  A characterization of {\em any}  bent RS function of degree 2 follows from \cite[Th. 3.7, p. 6]{ZG13}.  Examples show that \cite[Th. 3.1, p. 3]{ZG13} is false.

We will now look at rotation symmetric functions $f=\bigoplus\limits_i f_i$ of any degree where each of the $f_i$ is a  short-cycle MRS Boolean function.  We will prove that $f$ is bent only if $f_i=f_0$ for some $i$.  To do this we need the following result from \cite[pp. 218]{zheng}.

\begin{theorem}
\label{Jnonlin}
Let $f$ be a function on $\mathbb{V}_n$ and $J\subset \{0,1,2,\ldots,n-1\}$ such that $f$ does not contain any term $x_{j_1}\cdots x_{j_t}$ where $t>1$ and $j_1,\ldots,j_t\in J$.  Then the nonlinearity of $f$ satisfies $\mathcal{N}_f\leq 2^{n-1}-2^{s-1}$, where $s=|J|$.
\end{theorem}

\begin{theorem}
\label{shortbent}
Let $n=2m$, then $f_0=\bigoplus\limits_{i=0}^{m-1}x_ix_{m+i}$ is the only bent  short-cycle MRS Boolean function in $n$ variables. 
\end{theorem} 
\begin{proof}
Let $f$ have SANF $x_0x_{a_1}\cdots x_{a_{d-1}}$ and suppose $f$ is a short-cycle function.  Then the number of monomials in $f$ is $k$ where $1\leq k<2m$.  Thus,
\begin{equation*} 
f=x_0x_{a_1}\cdots x_{a_{d-1}}\oplus x_1x_{{a_1}+1}\cdots x_{a_{d-1}+1}\oplus\cdots\oplus x_{k-1}x_{a_1+k-1}\cdots x_{a_{d-1}+k-1}
\end{equation*} 
It is easy to see then that $\{0,a_1,\ldots,a_{d-1}\}=\{k,a_1+k,\ldots,a_{d-1}+k\}=\ldots=\{lk,a_1+lk,\ldots,a_{d-1}+lk\}$ where $lk<2m$.  
By Lemma \ref{Ondiv} we know $k$ must divide $n=2m$, thus $1\leq k\leq m$.  Thus by rotation symmetry the tuple $\{a_i,a_i+k,\ldots,a_i+lk\}$ appears in each monomial of $f$, for every $i,~0 \leq i \leq d - 1.$  If $k=m$ then $l=1$ and $f$ has SANF
\begin{equation}
\label{[4]}
f=x_0x_{a_1}x_{a_2}\ldots x_{a_{g-1}}x_mx_{m+a_1}\ldots x_{m+a_{g-1}}, 
\end{equation}  
where $0<a_1<a_2<\ldots<a_{g-1}<m$. First consider $g>1$.  Suppose $a_k<m$ and $a_k+1=m;$ then $m+a_k<2m$ and $m+a_k+1=2m$.  So if $x_{a_k+1}=x_{m}$ then $x_{m+a_k+1}=x_0$.  Thus for each monomial in $f,$ $g$ of the variables $x_i$ have $i<m$ and $g$ have $i \geq m$.  Let $J=\{m-1,m,\ldots,n-1\}$.  Thus by Theorem \ref{Jnonlin} we have $\mathcal{N}_f\leq 2^{n-1}-2^{m} < 2^{n-1}-2^{m-1},$ so $f$ is not bent by Theorem \ref{Had}.  If $g=1$ then $f=f_0$ which we know is bent.
Thus if $k=m$ then $f$ is bent only if $f=f_0$.  If $k<m$ then by regrouping the terms we have that,
\begin{equation*}
f=\bigoplus\limits_{i=0}^{k-1}x_ix_{b_1+i}\cdots x_{b_g+i}x_{k+i}x_{b_1+k+i}\cdots x_{b_g+k+i}\cdots x_{lk+i}x_{b_1+lk+i}\cdots x_{b_g+lk+i}
\end{equation*}  
From the first term in the monomial, $x_i$, we see that each monomial contains at least one element of the set $\{x_0,x_1,\ldots,x_{k-1}\}$.  Since $k<m$ then $k-1<m-1$.  So if we let $J=\{m-1,m,m+1,\ldots,n-1\}$, then by Theorem \ref{Jnonlin} we have $\mathcal{N}_f\leq 2^{n-1}-2^m<2^{n-1}-2^{m-1}$.  Thus $f$ is not bent by Theorem \ref{Had}.
\end{proof}

\begin{theorem}
\label{sumshortfncs}
Let $n=2m$ and let $f$ in $n$ variables be a sum of MRS short-cycle  functions, so $f=\bigoplus\limits_if_i$, where each $f_i$ has fewer than $n$ monomials.  If $f$ is bent then one of the $f_i$ is $f_0=\bigoplus\limits_{i=0}^{m-1}x_ix_{m+i}$. 
\end{theorem} 
\begin{proof}
Assume $f$ is bent and no $f_i$ is $f_0$.  Then any short-cycle of the form in equation \eqref{[4]} has $g>1$ and from the proof of Theorem \ref{shortbent} we see that each monomial in $f$ contains at least one variable $x_i$ where $i<m-1$.  Thus if we let $J=\{m-1,m,\ldots,n-1\}$ then $f$ does not contain any term $x_{j_1}\cdots x_{j_t}$ where $t>1$ and $j_i\in J$; hence from Theorem \ref{Jnonlin} we have $\mathcal{N}_f\leq 2^{n-1}-2^m<2^{n-1}-2^{m-1}$.  Thus $f$ cannot be bent, contradiction.  Therefore one of the $f_i$ is $f_0$.
\end{proof} 

\section{Homogeneous rotation symmetric bent functions when n=2p} 
\label{HomRSprime}

Using ideas similar to those in section \ref{HomRSany}, we can prove that the only homogeneous MRS bent function in $n=2p$ variables, where $p>2$ is prime, is $f_0=\bigoplus\limits_{i=0}^{p-1}x_ix_{p+i}$.  By a different method, this result was proven for any $n$ by Meng et al. \cite[Th. 11, pp. 1114]{Me10}

Notice we have already proven this for some cases in section \ref{HomRSany}.  For the remainder of the cases we must first further simplify \eqref{[2]}.  To do this we will need a few facts about the rotation symmetric equivalence classes of $\mathbb{V}_n$.

\begin{lemma}
\label{Ondiv}
$|O_n({\bf v}_i)|$ divides n.
\end{lemma}
\begin{proof}
$O_n({\bf v}_i)$ is the orbit generated by ${\bf v}_i\in\mathbb{V}_n$ under the action of $G$ where $G$ is the group of left cyclic shifts.  Thus $|O_n({\bf v}_i)|=|G:G_{{\bf v}_i}|$, where $G_{{\bf v}_i}=\{\rho\in G|\rho({\bf v}_i)={\bf v}_i\}$.   Since $|G|=n$ then $|O_n({\bf v}_i)|$ divides $n.$
\end{proof}

\begin{lemma}
$|O_n({\bf v}_i)|=2\iff {\bf v}_i\sim{\bf v}_{\frac{2^n-1}{3}}$.
\end{lemma}
\begin{proof}
$\Rightarrow$: $|O_n({\bf v}_i)|=2 \iff O_n({\bf v}_i)=\{(1,0,1,0,\ldots,1,0),(0,1,0,1,\ldots,0,1)\}$.  Let ${\bf v}_i=(0,1,0,1,\ldots,0,1) \Rightarrow i=1+2^2+2^4+\ldots+2^{n-2}\Rightarrow$
\begin{eqnarray*}
3i &=&(1+2)i=(1+2)(1+2^2+2^4+\ldots+2^{n-2})\\
&=&(1+2^2+2^4+\ldots+2^{n-2})+(2+2^3+2^5+\ldots+2^{n-1})\\
&=&1+2+2^2+2^3+2^4+\ldots+2^{n-1}
\end{eqnarray*}
$\Rightarrow {\bf v}_{3i}={\bf v}_{2^n-1}$ since ${\bf v}_{2^n-1}=(1,1,1,\ldots,1)$.  So we have $i=\frac{2^n-1}{3}$.  Since $(0,1,0,1,\ldots,0,1)\sim(1,0,1,0,\ldots,1,0)$ we have $(1,0,1,0,\ldots,1,0) \sim{\bf v}_{\frac{2^n-1}{3}}$.\\
$\Leftarrow$: Reverse the argument used above.
\end{proof}

\begin{theorem}
If $n$ is even and $\frac{n}{2}$ is an odd prime, then
\begin{eqnarray}
\label{[3]}
M_{0,2^n-1} &=& 2+2\hat{f}({\bf v}_0)\hat{f}({\bf v}_{2^n-1})+2n\sum\limits_{\substack{{\bf v}_k\in\mathcal{G}_n \\ |O_n({\bf v}_k)|=n \\ {\bf v}_k\nsim{\bf v}_{2^n-1-k}} }\hat{f}({\bf v}_k)\hat{f}({\bf v}_{2^n-1-k})\nonumber \\
& &+n\sum\limits_{\substack{{\bf v}_k\in\mathcal{G}_n \\ |O_n({\bf v}_k)|=\frac{n}{2}}}\hat{f}({\bf v}_k)\hat{f}({\bf v}_{2^n-1-k})+n\sum\limits_{\substack{{\bf v}_k\in\mathcal{G}_n \\ |O_n({\bf v}_k)|=n \\ {\bf v}_k\sim{\bf v}_{2^n-1-k}}}(\hat{f}({\bf v}_k))^2.
\end{eqnarray}
\end{theorem}
\begin{proof}
Let $n=2p$ where $p$ is an odd prime.  Then by Lemma \ref{Ondiv}, $|O_n({\bf v}_i)|=1,2,p,\text{ or } n$.

Let ${\bf v}_i\in\mathbb{V}_n$ such that $|O_n({\bf v}_i)|=1$, then ${\bf v}_i={\bf v}_0$ or ${\bf v}_{2^n-1}$ which appears as the first term of \eqref{[2]} with $k=0$.

If $|O_n({\bf v}_i)|=2$ then ${\bf v}_i={\bf v}_{\frac{2^n-1}{3}}$ and $2^n-1-\frac{2^n-1}{3}=2\frac{2^n-1}{3}$, so ${\bf v}_i\sim{\bf v}_{2^n-1-i}$ and we have one $2(\hat{f}({\bf v}_{\frac{2^n-1}{3}})^2$ term.  Since $\hat{f}({\bf v}_i)=\pm 1$ for any ${\bf v}_i\in\mathbb{V}_n$ then $(\hat{f}({\bf v}_i))^2=1$ for any ${\bf v}_i$.  So when $k=\frac{2^n-1}{3}$ in \eqref{[2]}, then $2\hat{f}({\bf v}_k)\hat{f}({\bf v}_{2^n-1-k})=2$.

Now suppose ${\bf v}_i=(a_0,a_1,\ldots,a_{n-1})\in\mathcal{G}_n$ and $|O_n({\bf v}_i)|=p$, then ${\bf v}_i'={\bf v}_i''$.  So to consider the cyclic shifts of ${\bf v}_i$ we need only consider the cyclic shifts of $(a_0,a_1,\ldots,a_{p-1})$. Notice ${\bf v}_{2^n-1-i}=(1+a_0,1+a_1,\ldots,1+a_{n-1})$ and we have $(1+a_0,1+a_1,\ldots,1+a_{p-1})=(1+a_{p},1+a_{p+1},\ldots,1+a_{n-1})$ so $|O_n({\bf v}_{2^n-1-i})|=p$ as well.  Since $p$ is odd $(1+a_0,1+a_1,\ldots,1+a_{p-1})$ will have a different number of 1's than $(a_0,a_1,\ldots,a_{p-1})$. Hence ${\bf v}_i\nsim{\bf v}_{2^n-1-i}$ since $(a_0,a_1,\ldots,a_{p-1})\nsim(1+a_0,1+a_1,\ldots,1+a_{p-1})$.  Thus from Lemma \ref{conjeq}, if $|O_n({\bf v}_i)|=p$, then grouping the corresponding terms in \eqref{[2]} by the equivalence class representatives gives $2p\hat{f}({\bf v}_i)\hat{f}({\bf v}_{2^n-1-i})=n\hat{f}({\bf v}_i)\hat{f}({\bf v}_{2^n-1-i})$.

If $|O_n({\bf v}_i)|=n$ and ${\bf v}_i\nsim{\bf v}_{2^n-1-i}$ then grouping the corresponding terms in \eqref{[2]} by the equivalence class representative gives the new coefficient $2n$.  If ${\bf v}_i\sim{\bf v}_{2^n-1-i}$ then this coefficient becomes $2\frac{n}{2}=n$.

Thus \eqref{[2]} reduces to the above equation.
\end{proof}

\begin{lemma} 
\label{fhat}
Let $n=2p$, where $p>2$ is prime. If $f$ is a rotation symmetric bent Boolean function in $n$ variables, then $\hat{f}({\bf v}_0)=-\hat{f}({\bf v}_{2^n-1})$.
\end{lemma} 
\begin{proof}
From Theorem \ref{Had} $f$ is bent if and only if $M=2^nI_{2^n}$.  Thus if $f$ is bent then $M_{0,2^n-1}=0$.  From \eqref{[3]} we see that in order for $M_{0,2^n-1}$ to be $0$ we need the $2$ which appears from the case ${\bf v}_k={\bf v}_\frac{2^n-1}{3}$ to cancel.  Let 
\begin{eqnarray*}
& &a=\hat{f}({\bf v}_0)\hat{f}({\bf v}_{2^n-1}),\ b=\sum\limits_{\substack{{\bf v}_k\in\mathcal{G}_n \\ |O_n({\bf v}_k)|=n \\ {\bf v}_k\nsim{\bf v}_{2^n-1-k}} }\hat{f}({\bf v}_k)\hat{f}({\bf v}_{2^n-1-k}),\\
& &c=\sum\limits_{\substack{{\bf v}_k\in\mathcal{G}_n \\ |O_n({\bf v}_k)|=\frac{n}{2}}}\hat{f}({\bf v}_k)\hat{f}({\bf v}_{2^n-1-k}), \text{and }d=\sum\limits_{\substack{{\bf v}_k\in\mathcal{G}_n \\ |O_n({\bf v}_k)|=n \\ {\bf v}_k\sim{\bf v}_{2^n-1-k}}}(\hat{f}({\bf v}_k))^2.
\end{eqnarray*}
Assume $f$ is bent, then \eqref{[3]} becomes:
\begin{equation*}
2+2a+4pb+2pc+2pd=0\Rightarrow p(2b+c+d)=-1-a.
\end{equation*} 
If $a=1$, then $p(2b+c+d)=-2$.  Since $p>2$ and $2b+c+d\in\mathbf{Z}$ this is a contradiction.  Thus $a=-1$ and $\hat{f}({\bf v}_0)\hat{f}({\bf v}_{2^n-1})=-1$.  Therefore $\hat{f}({\bf v}_0)=-\hat{f}({\bf v}_{2^n-1})$ since $\hat{f}({\bf v}_i)=\pm 1$ for all $i$.
\end{proof}

The following corollary now proves that full-cycle homogeneous MRS Boolean functions in $2p$ variables cannot be bent.  

\begin{corollary}
\label{hombent1}
Let $n=2p$ where $p$ is an odd prime.  Let $f$ be an MRS boolean function with SANF $x_0x_{a_1}x_{a_2}\ldots x_{a_{d-1}}$.  If the number of monomials in $f$ is $n$, then $f$ is not bent.
\end{corollary}
\begin{proof}
Using Theorem \ref{Had} we need only show $M_{0,2^n-1}\neq 0$.  If the number of monomials in $f$ is $n$, then we have $\hat{f}({\bf v}_0)=\hat{f}({\bf v}_{2^n-1})$ since $n=2p$ is even.  Thus $M_{0,2^n-1}\neq 0$ so $f$ is not bent.  
\end{proof} 

We can now prove that $f_0$ is the only homogeneous MRS bent function.

\begin{theorem}
\label{f0onlyS}
Let $n=2p$, where $p>2$ is prime.  The only bent homogeneous MRS Boolean function in $n$ variables is  
$f_0=\bigoplus\limits_{i=0}^{p-1}x_ix_{p+i}$. 
\end{theorem}
\begin{proof}
Any monomial homogeneous rotation symmetric function, $f$, has SANF $x_0x_{a_1}\ldots x_{a_{d-1}}$ and is either short-cycle or full-cycle.  If $f$ is full-cycle then from Corollary \ref{hombent1}  $f$ is not bent.  If $f$ is short-cycle then from Theorem \ref{shortbent}, $f$ is bent if and only if $f=f_0$.
\end{proof}  

\begin{theorem}
\label{short2form}
Let $n=2p$, where $p>2$ is prime, and let $f$ be a rotation symmetric Boolean function in $n$ variables with SANF $x_0x_{a_1}\cdots x_{a_{d-1}}$.  If $f$ has $p$ monomials, then $d$ is even  and $f$ has SANF given by equation \eqref{[4]}.
\end{theorem}
\begin{proof}
If the number of monomials in $f$ is $p$ then
\begin{equation*}
f=x_0x_{a_1}\cdots x_{a_{d-1}}\oplus x_1x_{a_1+1}\cdots x_{a_{d-1}+1}\oplus\cdots\oplus x_{p-1}x_{a_1+p-1}\cdots x_{a_{d-1}+p-1}
\end{equation*}
Thus $\{0,a_1,\ldots,a_{a_{d-1}}\}=\{p,a_1+p,\ldots,a_{d-1}+p\}$ so each pair $\{x_{a_i},x_{a_i+p}\}$ appears in every monomial and $f$ is of the form in \eqref{[4]}.  The result that $d$ is even follows immediately since each monomial contains the pairs $\{x_{a_i},x_{a_i+p}\}$.
\end{proof}

\begin{theorem}
\label{hombentfull} 
Let $n=2p$, where $p>2$ is prime.  Let $f=\bigoplus\limits_{i=1}^{s}f_i$ where each $f_i$ has SANF $x_0x_{a_{i_1}}\cdots x_{a_{i_{d-1}}}$ and $deg(f)=d$. Let $r$ be the number of $f_i$ that are short-cycle and $l$ be the number of $f_i$ that are long-cycle.  If $f$ is bent then $d$ is even.  Furthermore, if $d=2$ then $f$ contains $f_0$ and if $d\neq 2$ then $r\geq 1$ is odd and $l\geq 1$.
\end{theorem}
\begin{proof} 
Let $f=f_1\oplus f_2\oplus\cdots\oplus f_s$ where each $f_i$ has SANF $x_0x_{a_{i_1}}\ldots x_{a_{i_{d-1}}}$.  Thus $deg(f_i)=d\ \forall i$ and so $deg(f)=d$.  Since $n=2p$ then $\text{GCD}(n,d)=1,2,p,\text{or }n$ and from Lemma \ref{Ondiv} we know that the number of monomials in each $f_i$ is either $1,2,p,\text{or }n$.  
\begin{itemize}
\item[(i)] Suppose the number of monomials in $f_i$ for some $i$ is $1$, then since $f_i$ is rotation symmetric then $f_i=x_0x_1\ldots x_{n-1}$ and $d=n$. 
\item[(ii)] Suppose the number of monomials in $f_i$ for some $i$ is $2$, then $f=x_0x_{a_{i_1}}\ldots x_{a_{i_{d-1}}}\oplus x_1x_{a_{i_1}+1}\ldots x_{a_{i_{d-1}}+1}$.  It is easy to see by induction that the first monomial contains all of the even indices and the second monomial contains all of the odd indices.  Thus $d=p$. 
\item[(iii)] Suppose the number of monomials in $f_i$ for some $i$ is $p$, then from Theorem \ref{short2form} $f_i$ is of the form in equation \eqref{[4]} and $d$ is even.
\end{itemize} 
Thus $f_i$ is a short-cycle function only if $d=p,n,\text{or is even}$.  Let $\text{GCD}(n,d)=1$, then the number of monomials in $f_i$ is $n$ for all $i$.  Thus, $\hat{f_i}({\bf v}_0)=\hat{f_i}({\bf v}_{2^n-1})\ \forall i$ and so $\hat{f}({\bf v}_0)=\hat{f}({\bf v}_{2^n-1})$ and by Lemma \ref{fhat} $f$ is not bent.  If $\text{GCD}(n,d)=p$ then $f_i$ either has $2$ monomials or $n$ monomials for each $i$.  Thus $\hat{f}({\bf v}_0)=\hat{f}({\bf v}_{2^n-1})$ and so $f$ is not bent.  If $\text{GCD}(n,d)=n$ then $d=n$ and $f$ is not bent by Lemma \ref{degn/2}.  Now suppose $\text{GCD}(n,d)=2$, thus $d$ is even.  If $d=2$ then by Corollary \ref{2f0contain} $f$ must contain $f_0$ which is the only degree $2$ short-cycle.  If $d\neq 2$ then by (iii), the short-cycle functions of degree $d$ are of the form in equation \eqref{[4]}.  By Theorem \ref{sumshortfncs} we know that any combination of these short-cycle functions is not bent, thus $f$ must also contain a long-cycle and so $l\geq 1$.  If $r$ is even then,
\begin{equation*}
\hat{f}({\bf v}_0)=\prod\limits_{i=1}^{s}\hat{f_i}({\bf v}_0)=(-1)^r\prod\limits_{i=1}^{s}\hat{f_i}({\bf v}_{2^n-1})=\hat{f}({\bf v}_{2^n-1})
\end{equation*}
since $\hat{f_i}({\bf v}_0)=-\hat{f_i}({\bf v}_{2^n-1})$ when $f_i$ is short-cycle and $\hat{f_i}({\bf v}_0)=\hat{f_i}({\bf v}_{2^n-1})$ when $f_i$ is long-cycle.  Thus by Lemma \ref{fhat} $f$ is not bent.  
\end{proof}

\begin{remark}
The previous Theorem says that any homogeneous RotS bent function must have even degree and contain an odd number of short-cycle MRS functions.  If the degree is greater than $2$ then it must also contain at least one long-cycle MRS function.
\end{remark}

\begin{lemma}
\label{oddind}
Let $n=2p,$ $p$ an odd prime, and $f$ be a homogeneous RotS Boolean function in $n$ variables with $deg(f)=d=2k$ where $k>1$.  Let $f=\bigoplus\limits_{i=1}^{s}f_i\oplus\bigoplus\limits_{i=s+1}^{l}f_i$ where for all $1\leq i\leq s$, $f_i$ is an MRS short-cycle function and for all $s+1\leq i\leq l$, $f_i$ is an MRS full-cycle function.  Suppose $f_i$ has SANF $x_0x_{a_{i_1}}\ldots x_{a_{i_{d-1}}}$.  If $2\leq r\leq d-2$ for all $s+1\leq i\leq l$, where $r$ is the number of odd indices in $\{a_{i_1},\ldots,a_{i_{d-1}}\}$, then $f$ is not bent.
\end{lemma}
\begin{proof}
Since the degree of $f$ is even, we know from the proof of Theorem \ref{hombentfull} that any short-cycle is of the form in \eqref{[4]}.  Thus no short-cycle contains a term $x_{j_1}\ldots x_{j_t}$ where $t>1$ and $j_1,\ldots,j_t$ are in the set $J=\{0,2,4,\ldots,n-2,n-1\}$ since $deg(f)\geq 4$.  If the number of odd indices, $r$, in the first monomial of any  full-cycle   in $f$ is $2\leq r\leq d-2$, then each monomial in the full-cycles  has either $r$ odd indices and $d-r$ even indices or $d-r$ odd indices and $r$ even indices.  Since $2\leq r\leq d-2$ then $2\leq d-r\leq d-2$, thus each monomial contains at least $2$ odd indices.  Thus, because $J$ contains only one odd number, $n-1$, then no monomial, $x_{j_1}\ldots x_{j_t}$, in a long-cycle has $x_{j_1},\ldots,x_{j_{t}}\in J$ where $t>1$.  Thus by Theorem \ref{Jnonlin}, $\mathcal{N}_f\leq 2^{n-1}-2^{p}<2^{n-1}-2^{p-1}$ and so $f$ is not bent by Theorem \ref{Had}.
\end{proof}

\begin{remark}
Lemma \ref{oddind} says that if $f=\bigoplus\limits_{i}f_i$ is a homogeneous RotS Boolean function of even degree $d$, $d>2$, in $n=2p$ variables then one of the $f_i$ must have a monomial with exactly one odd index.
\end{remark}

Let $f$ be a homogeneous RotS function of degree $d$ with SANF $\bigoplus\limits_{1\leq i\leq s}\beta_i$, where $\beta_i=x_{k_0^{(i)}}x_{k_2^{(i)}}\cdots x_{k_{d-1}^{(i)}}$ and $k_0^{(i)}=0$ for all $i$.  Define a sequence $s_j^{(i)},j=1,2,\ldots,d$, by $s_j^{(i)}=k_{j}^{(i)}-k_{j-1}^{(i)}$ for $1\leq j\leq d-1$, and $s_d^{(i)}=k_0^{(i)}+n-k_{d-1}^{(i)}$.  Let $s_f$ be the largest distance between two consecutive indices in all of the monomials in $f$, thus $s_f=\max\limits_{i,j}s_j^{(i)}$.  Then from \cite[Th. 13, pp. 1116]{Me10} we have the following:

\begin{theorem}
\label{meng}
Let $f$ be a homogeneous RotS function with degree $d\geq 3$ in $n$ variables.  If $s_f\leq \frac{n}{2}$, then $f$ is not bent.
\end{theorem} 

We can use Theorem \ref{meng} to get a useful bound on the degree of any possible homogeneous RotS bent function in $2p$ variables:
\begin{theorem}
Let $f$ be a homogeneous RotS bent function of degree $d\geq 3$ in $n=2p$ variables, $p$ an odd prime. Then $d \leq (p+3)/2.$ 
\end{theorem}
\begin{proof}
Suppose $d > (p+3)/2.$ We choose $J = \{0, 1, 2, 4,\ldots, 2p-2\}$ (so $|J| = p+1$ and $J$ has $p$ even elements) in Theorem \ref{Jnonlin}.
The theorem applies since now no monomial $x_0x_{a_1}\ldots x_{a_{d-1}}$ can have all of its variables $x_j$ in $J$ and also have a gap of length $\geq  (n/2)+1$
in its index set (necessary for $f$ to be bent by Theorem \ref{meng}).   Suppose the monomial $x_0x_{a_1}\ldots x_{a_{d-1}}$ has all of its variables $x_j$ in $J$ and $d>\frac{p+3}{2}$.  If there is a gap in the indices greater than or equal to $p+1$, then at least $\frac{p-1}{2}$ of the indices in $J$ do not appear in the monomial.  Thus there are at most $p+1-\frac{p-1}{2}=\frac{p+3}{2}$ possible indices in $J$ which appear in the monomial.  Thus $d\leq\frac{p+3}{2}$, contradiction. Note, for example,  that if the indices are 
$0, 1, \underbrace{2, p+3}_{gap}\, p+5, \ldots, 2p-2$  then we have a gap of
length $p+1$ but only $(p+3)/2$ indices, contradicting our assumption about $d.$ Now Theorem \ref{Jnonlin} gives

$$\mathcal{N}_f \leq 2^{n-1}-2^{p} \leq 2^{n-1} -2^{n/2} < 2^{n-1} - 2^{n/2 - 1},$$
so $f$ is not bent by Theorem \ref{Had}, contradicting our hypothesis.  
\end{proof}  






\color{black}

\color{black}

\section{Nonhomogeneous rotation symmetric bent functions when n=2p}
\label{nonhomRS}

We have already shown in Theorem \ref{sumshortfncs} that any bent function composed only of short-cycle MRS functions must contain $f_0$.  

We can now extend the ideas used in section \ref{HomRSprime} to show that in most cases any bent rotation symmetric Boolean function in $n=2p$, $p>2$ prime, variables must contain $f_0$.

\begin{theorem}
Let $n=2p$ where $p$ is an odd prime.  Let $f=f_1\oplus f_2\oplus\ldots\oplus f_s$ be a rotation symmetric Boolean function where each $f_i$ has SANF $x_0a_{a_1}\ldots x_{a_{d_i-1}}$.  If $d_i=2$ or the number of monomials in each $f_i$ is either $2\text{ or }n$, then $f$ is bent only if $f$ contains $f_0=\bigoplus\limits_{j=0}^{p-1}x_jx_{p+j}$.  If $f$ contains an even number of $f_i$, where the number of monomials in $f_i$ is $p$  and all other $f_i$ contain $2$ or $n$ monomials, then $f$ is not bent.
\end{theorem}
\begin{proof}
Let $f=f_1\oplus f_2\oplus\ldots\oplus f_s$ where for each $f_i$, $d_i=2$ or the number of monomials in $f_i$ is $2\text{ or }n$.  Suppose $f$ is bent and does not contain $f_0$.  Then since $f_0$ is the only short-cycle RotS function with degree $2$ then each $f_i$ contains either $2$ or $n$ monomials.  Thus from the proof of Corollary \ref{hombent1}, $\hat{f_i}({\bf v}_0)=\hat{f_i}({\bf v}_{2^n-1})$ for all $i$.  Thus 
\begin{equation*}
\hat{f}({\bf v}_0)=\prod \hat{f_i}({\bf v}_0)=\prod \hat{f_i}({\bf v}_{2^n-1})=\hat{f}({\bf v}_{2^n-1})
\end{equation*} 
and from Lemma \ref{fhat} we see that $f$ cannot be bent which is a contradiction.  Thus $f$ must contain $f_0$.

Now suppose the number of monomials in $f_1,f_2,\ldots, f_{2r}$ is $p$ and\\ $f_{2r+1},f_{2r+2},\ldots,f_{s}$ contain either $2$ or $n$ monomials.  Then $\hat{f_i}({\bf v}_0)=-\hat{f_i}({\bf v}_{2^n-1})$ for $1\leq i\leq 2r$ and $\hat{f_i}({\bf v}_0)=\hat{f}({\bf v}_{2^n-1})$ for $2r+1\leq i\leq s$.  Thus, 
\begin{eqnarray*}
\hat{f}({\bf v}_0)&=&\prod\limits_{i=1}^{s}\hat{f_i}({\bf v}_0)=\prod\limits_{i=1}^{2r}-\hat{f_i}({\bf v}_{2^n-1})\prod\limits_{i=2r+1}^{s}\hat{f_i}({\bf v}_{2^n-1})=(-1)^{2r}\hat{f}({\bf v}_{2^n-1})\\
&=&\hat{f}({\bf v}_{2^n-1}).
\end{eqnarray*}
Thus from Lemma \ref{fhat}, $f$ cannot be bent.
\end{proof}

\begin{remark}
If $f_i$ has $1$ monomial, for some $i$, then $deg(f_i)=n$.  Thus $deg(f)=n$ and so $f$ is not bent from Lemma \ref{degn/2}.
\end{remark}

\begin{remark}
The previous Theorem says that any nonhomogeneous RotS bent function must contain an odd number of short-cycle MRS functions of even degree and at least one function which is a long-cycle MRS function.
\end{remark}

\end{document}